\documentclass[preprint,12pt,3p]{elsarticle}
\usepackage{amsthm}
\usepackage{amssymb}
\usepackage{amsmath}

\newtheorem{theo}{Theorem}[]
\newtheorem{lem}[theo]{Lemma}
\newtheorem{rem}{Remark}

\newtheorem{algo}{Algorithm}[]

\begin{document}
\begin{frontmatter}

\title{Efficient Modular Arithmetic for SIMD Devices}
\author{Wilke Trei}
\date{October 04 2012}
\ead{wilke.trei@uni-oldenburg.de}
\address{Department of Mathematics, University of Oldenburg, 
26111 Oldenburg, Germany, Phone: +494417983219}

\begin{abstract}
This paper describes several new improvements of modular arithmetic and how to exploit them in order to gain more efficient implementations of commonly used algorithms, especially in cryptographic applications. We further present a new record for modular multiplications per second on a single desktop computer as well as a new record for the ECM factoring algorithm. This new results allow building personal computers which can handle more than 3 billion modular multiplications per second for a 192 bit module at moderate costs using modern graphic cards. 
\end{abstract}

\end{frontmatter}

\section*{Introduction}
Parallelization of computational intensive algorithms has always been an important task in computational number theory. \\ This task becomes even more crucial during the last years, since the clock rate of ordinary processors stagnates. Therefore the chip vendors begin to rise the numbers of computational units on a single chip in order to keep the performance increases high. \\

Simultaneously graphic cards -- which classically have many computational units but lack of control flow units -- got more and more programmable. With the introduction of NVidia's CUDA technology and later the OpenCL platform, graphic cards came more and more into focus of programmers and security institutions, attracted by the high level of performance these chips may offer. \\

for efficient use of these chips, it is important to keep the chip's internal parallelization high because these devices often adapts one single instruction to multiple data (\emph{SIMD}). In this paper we present some algorithmic improvements to optimize modular arithmetic for this type of devices. This improvements are generic in the sense that they are neither specific for one special number theoretical algorithm, nor are they limited to SIMD use only. \\    
  
In order to test our improvements in practice, we applied them to an highly efficient version of Lenstra's ECM algorithm. Our implementation breaks the old record in terms of modular multiplication per second stated at Eurocrypt 2009 \cite{bernstein2008gpuecm} and SHARKS 2009 \cite{bernstein2009billionmulmod}. \\

We decided to use OpenCL for our implementation. Since OpenCL is a free standard and widely available, it is easy to keep compatibility with many soft- and hardware platforms.  A detailed description of OpenCL can be found in Section \ref{OpenCL_intro}. Descriptions of the ECM algorithm and our implementation are given in Section \ref{ecm_on_gpu}. 

\section{Heterogeneous Computing using OpenCL} \label{OpenCL_intro}
\subsection{Overview} 

OpenCL is an open programming model and standard for heterogeneous hardware platforms. Its first version was released in December 2008 by the Khronos Group \cite{khronos2008opencl} and is especially designed for parallel computations. The latest version of the standard was released in November 2011 \cite{khronos2011opencl}.\\

OpenCL is designed to offer a unified programming model for different hardware platforms. Like the NVidia CUDA platform OpenCL can for instance be used to program modern graphic cards. Furthermore it can be used to program ordinary x86 CPUs as well as IBM Cell Processors and several special purpose hardware. At the end of this section we give an overview on the commonly used OpenCL devices and its computational capabilities. We will describe the OpenCL programming platform briefly. An elaborate description can be found in \cite{gaster2011heterogeneous}. \\

The OpenCL platform model consists of two hardware components -- a \emph{host} and an \emph{OpenCL device} -- that may or may not be the same. While the host must be a directly programmable device like a CPU, the OpenCl device must not have a stand-alone functionality at all. An OpenCL device has its own memory location and offers several \emph{OpenCL Compute Units} (CU). A compute unit can be an SIMD (single instruction multiple data) computational unit, thus compute units are the finest granularity for control flow in the OpenCL platform. \\

A compute unit is furthermore the location of so called \emph{local memory} that can be used to share data among threads quickly. Every compute units may hold arbitrarily many \emph{stream cores} that are essentially arithmetical logical units (ALU). The stream core is the finest granularity for independent work threads, i.e. every stream core gets at least one ore more consecutive computational tasks. \\

A typical OpenCL program consists of two parts. The \emph{host code} is an ordinary program running on the host and is written in an arbitrary programming language. It binds the OpenCL libraries, loads work data and launches the so called \emph{OpenCL kernel}. An OpenCL kernel is a program that can be executed on a stream core. It describes what operations are applied to a single work item data. Kernels are usually executed in parallel on all available stream cores of an entire OpenCL device processing many work threads simultaneously. \\

The work items are normally grouped into so called \emph{work groups}. A single work group is atomically executed on one compute unit and can use the local memory of this unit to share data among its work items. To share data with all other threads the devices memory, also called \emph{global memory}, has to be used. \\

There are many different OpenCL devices on the market. The most commonly used are graphic cards. Table 1 gives a short overview of some common OpenCL devices and their computational capabilities. 
\begin{table}[h]  
\centering
\begin{tabular}{|l|c|c|c|}
\hline Vendor & Intel & NVidia & AMD \\ 
\hline Device & i7-3960X  & GTX 580 & HD 5870 \\
\hline Type & x86 CPU & GPU & GPU \\
\hline Compute Units & 6 / 12 & 16 & 20 \\ 
\hline Cores / CU & 1 & 32 & 16 \\ 
\hline Clock (MHZ) & 3300 & 1544 & 850 \\ 
\hline GFlops (SP) & 137.5 & 1581.1 & 2720 \\ 
\hline Global Memory & configurable  & 1 - 2 GB & 1 - 2 GB \\ 
\hline Local Memory & configurable  & 32 kbyte  & 32 kbyte \\ 
\hline Hardware Registers & 32 $\cdot$ 64 bit & 256 $\cdot$ 32 bit & 124 $\cdot$ 128 bit \\ 
\hline 
\end{tabular}
\caption{On Market OpenCL Devices} 
\end{table}
\subsection{Limitations}
The most common OpenCL devices are so called SIMD devices, i.e. devices that have many computational cores working on data in parallel, but sharing the instructional data and the control flow. There are several bottlenecks on this type of devices, especially on graphic cards. \\

One major aspect are the so called \emph{race conditions}. For example on a modern AMD graphic card at least 64 computational threads have to follow the same execution path independent of their data. Thus in case of branches, that are not coherent among all glued units, every occurring path of the branch has to be calculated sequentially while those units not taking this path remain idle. \cite[p. 135]{gaster2011heterogeneous}    \\

Another important aspect is the control of memory consumption and pressure on the memory bus system. While a single CPU has very few registers but several megabyte cache per compute core, a GPU has a lot of general purpose register space per thread available but only a very small cache, if any. For example one has 124 registers of 128 bit size available per thread for general purpose on a AMD Radeon HD 5000 series graphic card, but only constant data is cached within the L2 cache. \cite[appendix table D4]{amd2011guide} \\

There are two memory locations available for synchronizing work items. First of all there is the so called \emph{local memory} that offers fast access and high throughput. This local memory is placed within the compute unit and has a size of roughly 16-32 kb depending on the OpenCl device used. The local memory is designed to share data among all work items that belong to the same work group, i.e. are running on the same compute unit. \\

The other memory usable for synchronization is the global memory and can be up to a few gigabytes in size. This memory offers much slower access than local or register space, but can share data among all work items. Furthermore this memory location is the place where initial data and result data is stored. One crucial task in programming with OpenCl is to control the use of the global memory carefully, because it is easily going to be the bottleneck in any parallelized algorithm. \\

Even when all this limitations are considered, the programming itself is not as simple as on ordinary processors due to more architectural differences. This affects especially programming AMD graphic cards prior to the HD7000 series. On this cards every compute core itself is a vector processor able to handle up to 4 or 5 low level operations in parallel. For example on a HD5800 type card one core can perform a single integer multiplication and up to 4 independent integer additions in parallel. The need of splinting a single task into vectorized operations is one of the main difficulties when dealing with these graphic cards. In order to make programming more similar to ordinary CPU programming or working with NVidia graphic cards, AMD changed the architecture from the HD7000 series onwards to compute a single operation per core per clock. \cite[section 1.2]{amd2011guide} 

\section{Efficient Modular Arithmetic} \label{arithmetic}
\subsection{Common Modular Arithmetic} \label{common_arithmetic}
For most computationally difficult number theoretic algorithms it is needed to chain a lot of modular operations with fixed module. \\

Currently there exist two important algorithms for modular reduction using the fact that the modulus is fixed in most cryptographic applications, namely the Barret and the Montgomery reduction algorithms. \\

\begin{algo}{Barret Reduction \cite{barrett1987implementing}} ~ \\
Let $a,m \in \mathbb{N}$ with $a < m^2$ and $m$ be an odd integer of binary length $n = \lceil \mathrm{log}_2 \; m \rceil$. Furthermore, let $R = 2^n$ and $\mu = \lfloor \frac{R^2}{m} \rfloor$. The following algorithm computes $a \pmod{m}$. 
\begin{enumerate}
\item Calculate $r = a - \left\lfloor \lfloor \frac{a}{R} \rfloor \frac{\mu}{R} \right\rfloor m$.
\item Return $r-2m, r-m$ or $r$ depending which is in the range $[0,m[$.
\end{enumerate} 
\end{algo}

\begin{algo}{Montgomery Reduction \cite{montgomery1985modular} } ~\\ \label{montg_reduction}
Let $a,m \in \mathbb{N}$ with $a < m^2$ and $m$ be an odd integer of binary length $n = \lceil \mathrm{log}_2 \; m \rceil$. Furthermore, let $R = 2^n$ and $m' < R$ such that $m \cdot m' \equiv -1 \pmod{R}$. \\
The following algorithm computes $R^{-1} \cdot a \pmod{m}$. 
\begin{enumerate}
\item Calculate $b = a \cdot m' \pmod{R}$.
\item Calculate $r = \frac{a + b \cdot m}{R}$ over the integers.
\item If $r \geq m$ return $r-m$, else  return $r$.
\end{enumerate}

\end{algo}   

While the Barret Reduction gives the desired result immediately, the Montgomery Reduction is usually used with a modified residue system modulo $m$. In this case every element modulo $m$ is multiplied by $R$.  Using this transformation the modular addition is untouched and the multiplication can be done by multiplying $xR \cdot yR$ over the integers and then executing the Montgomery Reduction giving $xyR$.\\

Both reduction themselves cost at most $2M(n)$ if we define $M(n)$ to be the cost of a multiplication with input operands of size at most $n$ in terms of arithmetic operations. This claim holds, because the reduction modulo $R$ and the division operations are only binary representation cutoffs, and since $\mu$, $R$ and $m'$ can be pre-calculated. \\

Although the Montgomery reduction consumes more addition operations than the Barret reduction, we choose the latter algorithm for our implementation of modular reduction. This is especially due to the improvements provided in sections \ref{new_add} and \ref{new_mul}. For the rest of this paper we assume the size of the modulus $m$ is given by $n = \lceil \mathrm{log}_2 \; m \rceil$. \\

We recall that the cost of a modular multiplication is bounded by the cost of an ordinary multiplication of integers. Thus, it is crucial to know the integer multiplication methods when dealing with modular multiplication. \\

The classical schoolbook multiplication splits the input operands $a,b$ of size $n$ into two parts $a = a_1 \cdot 2^{\lceil \frac{n}{2} \rceil} + a_0, \; b = b_1 \cdot 2^{\lceil \frac{n}{2} \rceil} + b_0$ where $a_1, a_0, b_1, b_0$ are integers of size at most $\lceil \frac{n}{2} \rceil$. Then it performs the entire multiplication by calculating four products of half-size integers \begin{align*}
a b = a_1 b_1 2^{2\lceil \frac{n}{2} \rceil} + (a_1 b_0 + a_0 b_1) 2^{\lceil \frac{n}{2} \rceil} + a_0 b_0.
\end{align*}

This operation is used recursively until the machine word size is reached, i.e. the multiplications can be performed by the machine directly. The cost of this method is asymptotically  $\mathcal{O}(n^2)$. \\

With growing $n$ it becomes continuously harder to multiply in sufficiently short time. The first approach to decrease the complexity of the integer multiplication was due to A. Karatsuba and Y. Ofman \cite{karatsuba1963multiplication}. Based on the idea of Karatsuba and Ofman and the use of polynomial multiplication, A. Toom and S. Cook later developed a family of multiplication algorithms. The runtime of those algorithms depends on the family's parameter $k$ and a complexity class of $\mathcal{O}(n^\epsilon)$ can be obtained with $\epsilon$ arbitrarily close to $1$ for increasing $k$. The algorithm makes use of the evaluation homomorphism and the possibility to interpolate a polynomial of degree $k$ when $k+1$ points on its graph are known.

\begin{algo}{Toom-Cook-k \cite{cook1966minimum}} ~ \\ \label{toom-alg}
Let $a,b$ be integers of length $n$ and $k$ a fixed positive integer. Assume the binary representation of $a$ and $b$ is split into $k$ parts 
\begin{align*}
a &= \sum\limits_{i = 0}^{k-1} a_i \; 2^{i \cdot \lceil \frac{n}{k}\rceil} \\
b &= \sum\limits_{i = 0}^{k-1} b_i \; 2^{i \cdot \lceil \frac{n}{k}\rceil}.
\end{align*}
Then define the polynomials 
\begin{align*}
\bar{a} &:= \sum\limits_{i = 0}^{k-1} a_i \; x^{i} \in \mathbb{Z}[x] \\
\bar{b} &:= \sum\limits_{i = 0}^{k-1} b_i \; x^{i} \in \mathbb{Z}[x].
\end{align*}
Obviously $a = \bar{a}(2^{\lceil \frac{n}{k}\rceil})$ and $b = \bar{b}(2^{\lceil \frac{n}{k}\rceil})$ and due to the evaluation homomorphism $ab = (\bar{a} \cdot \bar{b})(2^{\lceil \frac{n}{k}\rceil})$. Since the product polynomial of $\bar{a}$ and $\bar{b}$ has degree $2k-2$, this product can be calculated as follows. \\

\begin{enumerate}
\item Select $2k - 1$ \emph{small} distinct integers $x_0, \hdots, x_{2k-2}$ and calculate the evaluations $y_i = \bar{a}(x_i)$ and $z_i = \bar{b}(x_i)$.
\item For all $0 \leq i < 2k-1$, calculate the products $w_i = y_i z_i$.
\item If the evaluation points $x_i$ are selected carefully, one can interpolate $\bar{a} \cdot \bar{b}$ from the known points $w_i = (\bar{a} \cdot \bar{b})(x_i)$ using linear algebra.
\item Evaluate $\bar{a} \cdot \bar{b}$ at $2^{\lceil \frac{n}{k}\rceil}$ to obtain the integer product of $a$ and $b$.
\end{enumerate}

\end{algo}

\begin{rem}{}
Usually the first one of the evaluation points in Algorithm \ref{toom-alg} is chosen to be $0$, so $w_0$ is essentially $a_0 b_0$. Furthermore one often assumes that the polynomials are evaluated at $x_{2k-2} = \infty$, i.e. their highest coefficients are multiplied. \\
Using this notation the Toom-Cook-2 algorithm with evaluation at $0, 1$ and $\infty$ is exactly the Karatsuba-Ofman algorithm.
\end{rem}

Algorithm \ref{toom-alg} divides $a$ and $b$ into $k$ parts and needs $2k-1$ multiplications beside several bit-shifts, multiplication with small constants etc. Thus the algorithm has an asymptotically complexity of $\mathcal{O}(n^{\operatorname{log}_k 2k-1})$ and hence for every $\epsilon > 1$, a complexity of $\mathcal{O}(n^\epsilon)$ can be achieved for sufficiently large $k$.  \\

Due to the overhead in steps 1,3 and 4 the Toom-Cook algorithms are only practical in a certain range. Usually after the inputs became 10-100 times the machine word size the arithmetic switches from schoolbook multiplication to the Karatsuba-Ofman algorithm. Later Toom-Cook-3 and Toom-Cook-4 become faster. Finally the Sch\"onhage-Strassen algorithm \cite{schoenhage1982asymptotically} -- that has currently fastest asymptotically complexity of an multiplication algorithm, i.e $\mathcal{O}(n \operatorname{log} n \operatorname{log} \operatorname{log} n)$ -- is more often used than Toom-Cook-k for $k$ exeeding $5$.

\subsection{Avoidance of Reduction Operations} \label{new_add}
As mentioned in the restrictions section it is necessary to avoid as many branches as possible and to keep the existing ones short. For modular arithmetic there is at least one barely avoidable branch dealing with reduction operations after over- or underflows. \\

On a SIMD device the branch times adds up, thus it does not matter if a reduction operation is performed every times when possible. The correct result simply can be selected afterwards. Therefore it is a common optimization to substitute the decision and reduction operations after additions or subtractions by the following algorithms.

\begin{algo}{Reduction after Addition} ~\\
Let $a$ be an integer with $0 \leq a < 2m$. Then the reduction of $a \pmod{m}$ can be computed by \begin{enumerate}
\item Calculate $a' = a-m$.
\item If $ a' < 0$ return $a$, else return $a'$.
\end{enumerate}
\end{algo}

\begin{algo}{Reduction after Subtraction} ~\\
Let $a$ be an integer with $-m \leq a < m$. Then the reduction of $a \pmod{m}$ can be computed by \begin{enumerate}
\item Calculate $a' = a+m$.
\item If $ a < 0$ return $a$, else return $a'$.
\end{enumerate}
\end{algo}

These algorithms have the advantage of very simple decisions depending only on a single bit. Furthermore they are simple to implement using OpenCl, since OpenCl offers very efficient selection operations and thus they cause only very short branches. \\

For our implementation we also use an observation on the Montgomery reduction that helps avoiding reduction operations after every multiplication step. 

\begin{lem}{Reduction Capabilities of the Montgomery Reduction} ~\\
Let $m \in \mathbb{N}$ be an odd integer of binary length $n = \lceil \operatorname{log}_2 \; m \rceil$ and $R, R'$ being powers of $2$ with $R \geq 4R' = 2^n$. Then the following observations can be made for the Montgomery Reduction algorithm ($\operatorname{redc}$) applied to $R$ and $m$. \begin{enumerate}
\item Let $a,b$ be integers with $a,b < 2R'$. Then $\operatorname{redc}(a \cdot b) \leq R' + m < 2 R'$.
\item Let $a,b$ be integers with $a,b < 3R'$. Then $\operatorname{redc}(a \cdot b) < \frac{13}{4} R'$
\end{enumerate}
\end{lem}  

\begin{proof}
The product of $a$ and $b$ is less or equal to $4{R'}^2$ which is bounded by $R R'$. Due to the function principle of the REDC algorithm this implies that
\begin{align*} \operatorname{redc}(a,b) = \frac{m N + ab}{R} < \frac{R N + R R'}{R} = R' + N.
\end{align*}
The second claim can be proven analogously.
\end{proof}

The Lemma allows to build multiplication chains without any intermediate reduction. 
In order to use this improvement, we ensure in our implementation, that our total work width is at least two bit wider than our modulus. That is why our OpenCL implementation of the ECM algorithm is currently limited to use integers up to $2^{254}$ instead of $2^{256}$. Furthermore we decided to allow representatives between $0$ and $2^{log_2 \; m \; +1} -1$ instead of $0$ and $m$ where $m$ is our modulus. Hence it may be required to add $2m$ instead of $m$ after a subtraction. Since after nearly every subtraction a multiplication follows, it causes no disadvantage to add the precomputed $2m$. \\

In this fashion we are able to limit the \emph{Reduction after Addition/Subtraction} operations from total 19 to 11 executions within our algorithms main loop. 

\subsection{Faster Truncated Multiplication} \label{new_mul}

One key task of both reduction algorithms described in Section \ref{common_arithmetic} is to use truncated multiplications, i.e. to perform multiplications where only one half of the results binary representation is needed for further use. In order to make the reduction operations as cheap as possible, it is a natural attempt to use the fact that one half of the product may be easier to calculate. For the lower half multiplication result there are well known methods to compute it in less time than the time for the full product. \\

\begin{algo} ~\\ \label{lower_half_product}
Let $a,b$ be integers of length $n$ and let $\rho$ be a parameter in the interval $[0.5, 1]$. Then one can calculate $a \cdot b \pmod{2^n}$ the following way.
\begin{enumerate}
\item If $a = \sum\limits_{i=0}^{n-1} a_i 2^i$ and $b = \sum\limits_{i=0}^{n-1} b_i 2^i$ calculate $P_0 = \sum\limits_{i=0}^{\lceil \rho (n-1) \rceil} a_i 2^i \cdot \sum\limits_{i=0}^{\lceil \rho (n-1) \rceil} b_i 2^i$.
\item Calculate the lower halves of the products $P_1, P_2$ satisfying \begin{align*}
P_1 &= \sum\limits_{i= \lceil \rho (n-1) \rceil}^{n-1} a_i 2^i \cdot \sum\limits_{i=0}^{n-1- \lceil \rho (n-1) \rceil} b_i 2^i \\
P_2 &= \sum\limits_{i=0}^{n-1- \lceil \rho (n-1) \rceil} a_i 2^i \cdot \sum\limits_{i= \lceil \rho (n-1) \rceil}^{n-1} b_i 2^i.
\end{align*}
\item Calculate the sum $a \cdot b \equiv P_0 + (P_1 + P_2)2^{\lceil \rho n \rceil} \pmod{2^{\lceil n \rceil} }$
\end{enumerate} 
\end{algo}

The case $\rho = 1$ in Algorithm \ref{lower_half_product} is exactly the case where the full product is calculated during step 1 and afterwards is reduced mod $2^n$. The case $\rho = 0.5$ is the average case when dealing with the classical schoolbook multiplication. In this case $P_0$ is equal to $a_0b_0$ and $P_1$ respectively $P_2$ are the products $a_1 b_0$ and $a_0 b_1$. The optimal choice of $\rho$ depends on the efficiency of the underlying multiplication algorithm. \\

\begin{theo}{Optimal $\rho$ for Algorithm \ref{lower_half_product} \cite[Section 2]{mulders97oncomputing} } ~ \\
Given a multiplication algorithm with complexity $\mathcal{O}(n^\alpha)$, $\alpha \in \; ]1,2]$. Then the runtime $S(n)$ of Algorithm \ref{lower_half_product} is at most \begin{align*}
S(n) &= M(\rho n) + 2 S ((1-\rho) n) \leq C_\rho M(n), \; \text{with} \\
C_\rho &= \frac{\rho^\alpha}{1-2(1-\rho)^\alpha}.
\end{align*}
The minimum value of $C_\rho$ is reached at  \begin{align*}
\hat{\rho} = 1- 2^\frac{-1}{\alpha-1}.
\end{align*}
\end{theo}

It is easy to see that for decreasing $\alpha$ the quantity $C_\rho$ of Algorithm \ref{lower_half_product} is tending to one. Thus a more efficient multiplication algorithm is less advantageous to calculating only half products. Table 2 gives some values for $\hat{\rho}$ and $C_\rho$ in the case of the multiplication algorithms discussed in Section \ref{common_arithmetic}. 
\begin{table}[h]
\centering
\begin{tabular}{|l|c|c|c|}
\hline  & Complexity Class & $\hat{\rho}$ & $C_\rho$ \\ 
\hline Schoolbook & $\mathcal{O}(n^2)$ & 0.500 & 0.500 \\ 
\hline Karatsuba-Ofman & $\mathcal{O}(n^{1.585})$ & 0.694 & 0.808 \\ 
\hline Toom-Cook-3 & $\mathcal{O}(n^{1.465})$ & 0.775 & 0.888 \\ 
\hline Toom-Cook-4 & $\mathcal{O}(n^{1.404})$ & 0.820 & 0.923 \\ 
\hline 
\end{tabular} 
\caption{$\hat{\rho}$ and $C_\rho$ for several multiplication algorithms}
\end{table}

Due to the missing carry bits from the lower half product, a method similar to Algorithm \ref{lower_half_product} is hard to achieve for the more significant bits of a multiplication. First attempts where to calculate some additional \emph{guard bits} to correct the error of missing carry bits \cite{hars2005fast}. \\

As an alternative Bentahar and Smart suggested an procedure called \emph{wooping} \cite{bentahar2007efficient}. The idea is that most calculations in the reduction algorithm are calculations over the integers and thus the error can be measured by doing analogous calculations modulo a small prime bigger than the maximal error. Using the idea of wooping one can obtain the same asymptotic speed for the calculation of high half truncated products as for the low half ones by costs of certain overhead. \\

For our purpose we found an astonishingly simple to use way to save operations in the calculation of the high half truncated product in step $2$ of the Montgomery Reduction algorithm. \\

In the setting of the Montgomery multiplication it is necessary to calculate the most significant bits of the integer $a+b \cdot m$. By proving the correctness of the Montgomery reduction algorithm one already knows that $b \cdot m \equiv -a \pmod{R}$ for any input number $a$. This fact can be used to save operations during the calculation of $b \cdot m$, since the lower half of its binary representation is already known. \\
We will demonstrate this advantage in the case of the schoolbook multiplication. \\

\begin{theo}{} ~ \\ \label{opt_schoolbook}
Let $n, M(n)$ be defined as before. Then the Montgomery Reduction algorithm for a modulus $m$ of size $n$ can be computed in $1.25 M(n)$ plus a few addition operations without the need of wooping or guard bits.
\end{theo}

\begin{proof}
Let the lower half product $b = a m' \pmod{R}$ already be calculated and let $m = m_1 2^{\lceil \frac{n}{2} \rceil} + m_0$ and $b = b_1 2^{\lceil \frac{n}{2} \rceil} + b_0$ be the split binary representations of $m$ and $b$ respectively. Then we know that
\begin{align*}
b \cdot m = m_1 b_1 2^{2\lceil \frac{n}{2} \rceil} + (m_1 b_0 + m_0 b_1) 2^{\lceil \frac{n}{2} \rceil} + m_0 b_0 \equiv -a \pmod{R}
\end{align*}
from the definition of the schoolbook multiplication and our observation. We assume \\ $R = 2^{2\lceil \frac{n}{2} \rceil}$ holds since $n$ is often ceiled to the next multiple of the machine word size, thus we can calculate
\begin{align*}
m_0 b_0 \equiv -(a + (m_1 b_0 + m_0 b_1) 2^{\lceil \frac{n}{2} \rceil}) \pmod{R}
\end{align*}
instead of multiplying $m_0b_0$ directly. This way one saves one of the four multiplications giving $0.5 M(n) + 0.75M(n)$ as approximate total runtime.
\end{proof}

This method is theoretically slower than wooping-based attempts for schoolbook multiplications. Since there is no need for an error correction it may be more practicable for small operand sizes. For growing $n$ one can even outperform whooping based attempts.

\begin{theo}{} \label{simple_opt}
Let the notation from Algorithms \ref{montg_reduction} and \ref{toom-alg} be given. Then one easily can calculate the full product in step 2 of Algorithm \ref{montg_reduction} using $2k-2$ instead of $2k-1$ sub-multiplications. 
\end{theo} 

\begin{proof}
The idea is to avoid the calculation of the lowest significant bits in the binary representation of $b \cdot m$. In detail the calculation of $\bar{b}(0) \bar{m}(0)$ can be replaced using the following algorithm.
\begin{enumerate}
\item Assume $x_0 = 0$ and $w_i = \bar{b}(x_i) \bar{m}(x_i)$ were calculated for all $i > 0$ using $2k -2$ multiplications.
\item Compute $w_{0,L} = -a \pmod{2^{\lceil \frac{n}{2} \rceil}}$. Since $w_0= \bar{b}(0) \bar{m}(0)$ is an integer less than $2^{2\lceil \frac{n}{2} \rceil}$, $w_{0,L}$ is exactly the lower half of the binary representation of $w_0$.
\item Use $w_1, \hdots, w_{2k-2}$, $w_{0,L}$ in order to compute the lower order digits $l_0$ of the linear term of the polynomial $\overline{b \cdot m}$. This works because the full linear term can be obtained from the full representations of $w_0, \hdots, w_{2k-2}$. 
\item Now one can recover the full product $\bar{b}(0) \bar{m}(0) = -a - l_0 2^{\lceil \frac{n}{2} \rceil} \pmod{2^{2\lceil \frac{n}{2} \rceil}}$. 
\end{enumerate}
\end{proof}

Table 2 shows the theoretical relative runtime $\hat{C} M(n)$ of the methods described in Theorems $\ref{opt_schoolbook}$ and $\ref{simple_opt}$ compared to the classical approach combined with \emph{guard bits} or \emph{wooping}. \\

\begin{table}[h] 
\centering
\begin{tabular}{|l|c|c|}
\hline  & $C_\rho$ & $\hat{C}$\\ 
\hline Schoolbook &  0.500 & 0.750\\ 
\hline Karatsuba-Ofman &  0.808 & 0.667\\ 
\hline Toom-Cook-3 &  0.888 & 0.800\\ 
\hline Toom-Cook-4 & 0.923 & 0.857\\ 
\hline 
\end{tabular} 
\label{theoretical_improvements}
\caption{Wooping compared to theorems \ref{opt_schoolbook} and \ref{simple_opt}}
\end{table}

Note that this optimizations makes the computation of our higher half truncated products faster than for the lower half ones. Furthermore from Toom-3 onwards there are more unused bits that may be used for further optimizations. This is one of the subjects of future work on this topic.

\section{ECM on Graphic Cards} \label{ecm_on_gpu}

For the evaluation of our arithmetical ideas we choose the ECM algorithm as a testing application. The ECM algorithm was first described by H.W. Lenstra Jr. in 1987 \cite{lenstra1987factoring} and was later improved by P. L. Montgomery \cite{montgomery1992fft} and many others. A single staged version of the algorithm is stated below.

\begin{algo}{Elliptic Curve Method, Stage 1 \cite{lenstra1987factoring}} ~ \label{ecm_algorithm} \\
Let $n$ be a composite integer and $B_1$ a suitable bound. The following probabilistic algorithm can be used to find a factor of $n$. \begin{enumerate}
\item Calculate the constant $k=\prod\limits_{p \in \mathbb{P}, p \leq B_1} p^{e_p}$ with $p^{e_p} \leq B_1 < e^{e_p +1}$, $e_p \in \mathbb{N}$.
\item Pick a random elliptic curve $E$ over $\mathbb{Z} / n \mathbb{Z}$ and a point $\mathcal{O} \neq P \in E_{\mathbb{Z} / n \mathbb{Z}}$.
\item Calculate $kP$ on $E_{\mathbb{Z} / n \mathbb{Z}}$. If this fails one denominator within the group law formulas is not invertible modulo $n$ and its gcd gives a factor of $n$. If no factor is found return to step 2 and pick a new curve.
\end{enumerate}
\end{algo}

There are several ways to implement this algorithm. One important choice is the model of the used curves. A common way is to use projective curves in Montgomery coordinates \cite{montgomery1987speeding} or Edwards coordinates \cite{bernstein2008edwards}. \\

Although elliptic curves in Edwards coordinates may have a slightly more efficient group law we decided to use Montgomery curves for our first prototype in order to keep compatibility with common implementations, especially the Brent-Suyama parametrization \cite{zimmermann2006ecm20}. Previous implementations of the ECM algorithm on graphic cards by Bernstein et al. used floating point arithmetic \cite{bernstein2008gpuecm} or $24$ bit multiplication \cite{bernstein2009billionmulmod} for building their long integer arithmetic. Although even for our used AMD HD 5000 series graphic cards the $24$ bit performance is roughly 5 times higher than for the ordinary $32$ bit multiplication, we decided to use the latter to build up our arithmetic. This is done because in contrast to the $24$ bit case OpenCL offers an easy to use command to get the higher order bits of an $32$ bit multiplication. Furthermore during a $32$ bit multiplication several additions can be performed in parallel, hence our code is designed to process carry bit calculations of previous multiplications in parallel to our current ones.  \\

The parallelization is straight forward. Every work-item is assigned to a single elliptic curve plus a starting point and has to process the scalar multiplication of the point described in step three of algorithm \ref{ecm_algorithm}. Since we build our program to handle $254$ bit integers -- $8 \cdot 32$ minus the two bits described in section \ref{new_add} -- and every item has to store the curve parameters, three copies of point coordinates and some temporary space, we are using roughly half of the $124$ registers available. Thus if one is careful with scratch space even 510 bit arithmetic should work the same way. Building wider arithmetic is one of the subjects for further improvements of our implementation.  \\

An important fact of this simple parallelization is the absence of need for synchronization between two threads. Even if not crucial, this helps a lot exploiting the maximum modular multiplication per second capabilities of modern graphic cards. Beside the arithmetic described in Section \ref{arithmetic} we use a carry select circuit on the low level. \\

An early version of the program lacking the improvements in Sections \ref{new_add} and \ref{new_mul} and several other improvements won a prize for innovative use of OpenCl assigned by AMD and TopCoder \cite{amd2011coding}.

\section{Experimental Results}
The tests of our implementation have been done on a standard personal computer using a Intel Core 2 Duo E8400 and 4 gigabytes of RAM. During the development we used an AMD Radeon HD5770 graphic card and for final measurements an AMD Radeon HD5870 type card. Note that the HD5770 has exactly half the compute units of the HD5870, thus it was easy to estimate the performance on a high end graphic card without using one. All our tests were performed using the AMD Catalyst Driver Version 11.08 among with the AMD APP SDK 2.5 on Ubuntu 11.10 x64 as our software platform. \\

All tests have been done on a 254 bit module and $B_1$ fixed to $8192$. For a better comparison to previous attempts to implement the ECM algorithm we scaled the results by $(\frac{254}{192})^2$. Table 4 shows this scaled results as well as the prize / performance ratio for our implementations. In order to compare the prize / performance ratio it is necessary to know that the previous records were achieved by using a 500\$ NVidia Geforce GTX295 graphics card while our HD5870 has an average market prize of 320\$ as of January 2012. 

\begin{table}[h]
\centering
\begin{tabular}{|c|c|c|}
\hline  & MulMod$\cdot 10^6$ / Sec & MulMod$\cdot 10^6$ / (Sec $\cdot \;$ \$) \\ 
\hline Previous Record \cite{bernstein2009billionmulmod} & 575 & 1.15 \\ 
\hline AMD Challenge Version  \cite{amd2011coding} & 430.8 & 1.34 \\ 
\hline Recent Optimized Version & 846.2 & 2.64 \\ 
\hline 
\end{tabular} 
\caption{192 bit MulMod / sec in our implementations}
\vspace{-1em}
\end{table}

Note that our implementation currently is able to handle 3756 scalar multiplications per second on an elliptic curve over a 254 bit modulus. Scaled to 192 bit these are roughly 6575 curves per second. Although this also is a new record for ECM on a single-chip graphics card it is not as much as it may be. The program of Bernstein et. al. \cite{bernstein2009billionmulmod} is capable of processing 4928 scalar multiplications per second by much less modular multiplications. We believe this gap occurs since we do not multiply out our module and do not use the arithmetical advantages of Edwards curves yet. \\

In order to evaluate the impact of our observations described in Sections \ref{new_add} and \ref{new_mul} we ran about 51.200 scalar multiplications with $B_1 = 8192$ and on a 254 bit modulus with OpenCL kernels without these improvements. Note that we used an AMD Radeon HD 5770 graphic card for our experiments that has exactly half the computational capabilities of the earlier described HD 5870 model.

\begin{table}[h]
\centering
\begin{tabular}{|c|c|c|c|}
\hline  & Curves / Sec & MulMod$\cdot 10^6$ / Sec & Ratio \\ 
\hline Without Optimizations & 1691.5 & 217.7 & 100 \%  \\ 
\hline Section \ref{new_add} only & 1744.6 & 224.6 & 103,1 \%  \\ 
\hline Section \ref{new_mul} only & 1818.6 & 234.1 & 107,6 \%  \\ 
\hline Fully Optimized & 1880.1 & 242.0 & 111,2 \% \\ 
\hline 
\end{tabular} 
\caption{254 bit MulMod / sec for different optimization stages on HD 5770}
\vspace{-1em}
\end{table} 
We see that our observations can deliver up to 11\% more arithmetical throughput for free, while being easy to implement. \\

By our experimental results we also see that even on a single chip graphics cards the limit of one billion modular multiplications for 192 bit modulus is in reach. To be more precise using two cards of type AMD Radeon HD 5970 -- that is two HD 5870 typed chips glued together at slightly lower clock rate -- the bar can be raised to about 3 billion modular multiplications per second. Setting up this kind of record breaking system costs two times 700\$ for the two graphic cards and additionally a few hundred \$ for the peripheral components. In total a total cost of 2000\$ should not be exceed. \\

This level of computational throughput is already in the range of being relevant for modern security. For example the elliptic curve discrete logarithm problem with a 130 bit module -- as used in one of the Certicom ECC challenges -- was believed to be infeasible for long. Using our 400 of our described 3 billion modular multiplication per second computers, this challenge is theoretically in range to be broken within a year under common runtime assumptions \cite{lenstra2000selecting}. \\

Summarizing SIMD devices offer a lot of computational potential and significance whereas there is still room left for more technical and algorithmical improvements.

\bibliography{biblio}{}
\bibliographystyle{elsarticle-num}

\end{document}